\documentclass[a4paper,12pt,reqno,superscriptaddress]{revtex4}

\usepackage[centertags]{amsmath}
\usepackage{amsfonts}
\usepackage{amssymb}
\usepackage{amsthm}
\usepackage{newlfont}
\usepackage{stmaryrd}
\usepackage{mathrsfs}
\usepackage{euscript}

\usepackage{color}

\theoremstyle{plain}
  \newtheorem{theorem}{Theorem}[section]

\theoremstyle{definition}

\theoremstyle{remark}
  \newtheorem{remark}[theorem]{Remark}

\numberwithin{equation}{section}

\let\al=\alpha \let\be=\beta \let\de=\delta 
\let\ve=\varepsilon  \let\ga=\gamma 
\let\ka=\kappa \let\la=\lambda \let\om=\omega 
\let\si=\sigma

   \let\Om=\Omega

\newcommand{\caP}{{\mathcal P}}

\newcommand{\caR}{{\mathcal R}}

\newcommand{\caT}{{\mathcal T}}

\newcommand{\opunit}{\text{1}\kern-0.22em\text{l}}
\newcommand{\funit}{\mathbf{1}}

\newcommand{\frj}{{\mathfrak j}}

\DeclareMathAlphabet{\mathpzc}{OT1}{pzc}{m}{it}

\newcommand{\id}{\textrm{d}}

\newcommand{\irr}{_{\text{IRR}}}

\begin{document}

\title{Rigorous meaning of McLennan ensembles}

\author{Christian Maes}
\affiliation{Instituut voor Theoretische Fysica, K.U.Leuven,
Belgium} \email{christian.maes@fys.kuleuven.be}
\author{Karel Neto\v{c}n\'{y}}
\affiliation{Institute of Physics AS CR, Prague, Czech Republic}
\email{netocny@fzu.cz}


\begin{abstract}
We analyze the exact meaning of expressions for nonequilibrium
stationary distributions in terms of entropy changes. They were
originally introduced by McLennan for mechanical systems close to
equilibrium and more recent work by Komatsu and Nakagawa has shown
their intimate relation to the transient fluctuation symmetry.
Here we derive these distributions for jump and diffusion Markov
processes and we clarify the order of the limits that take the
system both to its stationary regime and to the
close-to-equilibrium regime. In particular, we prove that it is
exactly the (finite) transient component of the irreversible part
of the entropy flux that corrects the Boltzmann distribution to
first order in the driving. We add further connections with the
notion of local equilibrium, with the Green-Kubo relation and with
a generalized expression for the stationary distribution in terms
of a reference equilibrium process.
\end{abstract}

\maketitle

\section{Introduction}
According to McLennan~\cite{mac,mac1}, the stationary density of
an open mechanical system away but close to thermal equilibrium
can be written in the modified Gibbs form
\begin{equation}\label{mac}
  \rho(x) \simeq \frac{1}{Z}\,e^{-\be H(x) + W(x)}
\end{equation}
with the nonequilibrium correction $W$ directly related to the
entropy production or to the dissipation in the driven system.
An essential feature of the formula, not quite visible yet, is
that the distribution $\rho$ is described in terms of macroscopic
parameters only, such as external temperature and driving fields.
It was also expected ``that some formal advantages may be offered
by an approach to nonequilibrium phenomena in which the Gibbs
ensemble plays a more prominent role'' (from the second paragraph
in \cite{mac}). Because of the suggested physical interpretation,
this proposal opens the possibility to construct nonequilibrium
statistical ensembles based on meaningful physical quantities,
see also \cite{LB,leb} for older and \cite{TM} for more recent work.
However, \eqref{mac} in \cite{mac} being just the
result of a formal perturbation calculation (together with some projection techniques), there have remained a number of difficulties with the exact meaning of this proposal as well as with its scope of generality. We mention some of these problems, as are clarified in the present paper:

(1) What entropy production does the correction term $W$
represent?---It can only be some transient component of the total
entropy production as the total entropy production clearly
diverges in the long-time (stationary) limit. Moreover, this
divergence (equal to the steady entropy flux) is in fact of order
$O(\ve^2)$ in some `distance of equilibrium' $\ve$ since it comes
from the product of thermodynamical forces and fluxes, both being
$O(\ve)$. The point will be that the transient irreversible part
of the entropy production and its linear part are both finite:
they coincide up to $O(\ve)$ and give a valid first-order
correction to the Boltzmann distribution as in~\eqref{mac}.

(2) Is the proposal also valid on different levels of description
than for mechanical systems?---The formal perturbation approach as
in~\cite{mac} does not reveal the essence and the physical
generality of the proposal.  Here the insight comes from dynamical
fluctuation theory: formula \eqref{mac} is basically a consequence
of the transient fluctuation symmetry, \cite{KN}. In other words,
it follows from the local detailed balance assumption which
determines the time-antisymmetric structure of the space-time
distribution in terms of the  history-dependent entropy fluxes,
cf. \cite{ls,mn}.  We make that visible for Markov processes.

(3) Can one go beyond close-to-equilibrium?---As explained
in~\cite{KN} and further applied in~\cite{sa,ha}, one can in
principle obtain a formal perturbation series for the stationary
distribution based on (all) the cumulants of the transient entropy
production. An interpretation has been given for the second-order expansion
where the divergences cancel out by a different way than explained in the present   paper. We instead present a generalization involving the so called dynamical activity or traffic, an advantage being that the evaluation is now from the start to be done under a reference equilibrium process.\\

We consider Markov processes of two types, jump and diffusion
processes.  Yet for simplicity we reserve the next section to
Markov jump processes; the diffusion case is formally completely
similar. We explain the relation with  local equilibrium in
Section \ref{leq}. In Section \ref{dif} we make some specific
remarks on the diffusion case and we also take there the
opportunity to illustrate an alternative to the derivation in
Section \ref{der}.  As an application, the Green-Kubo relations
are derived in Section \ref{gk}.  An illustration of the
McLennan-algorithm for an underdamped case is given in Section
\ref{odd}. We end, in Section \ref{max}, with a generalization
away from equilibrium.

\section{Markov jump processes}

After introducing some notation and basic concepts in the case of
jump processes, we give our main result which is a rigorous
version of the McLennan formula. A comparison to another approach
and further remarks are added.

\subsection{Set-up and assumptions}\label{s1}

Consider a continuous time Markov proces $x_t, t\geq 0$, taking
values in a finite state space $\Omega \ni x,y,\ldots$.  The
transition rates are $\lambda(x,y)\geq 0$ for jumps between the
states $x\rightarrow y$.  The Master equation for the probability
$\mu_t(x)$ of state $x$ as function of time $t$ is
\begin{equation}\label{timev}
  \frac{\id \mu_t(x)}{\id t} =
  \sum_{y \neq x} \{\mu_t(y) \la(y,x) - \mu_t(x) \la(x,y)\}
\end{equation}
with some given initial law $\mu_0=\mu$ at time zero.

Physical input distinguishes between equilibrium and
nonequilibrium dynamics. An equilibrium dynamics (with subscript
$0$) satisfies the condition of detailed balance, i.e.,
\begin{equation}\label{db}
\rho_0(x)\, \lambda_0(x,y) = \rho_0(y)\,\lambda_0(y,x)
\end{equation}
where $\rho_0(x) \propto e^{-\beta U(x)}$ for some potential $U$
and inverse temperature $\beta$, is then stationary. This relation
expresses the time-reversibility of the stationary equilibrium
process.  For nonequilibrium systems, detailed balance \eqref{db}
gets broken. An extension is known as the condition of local
detailed balance. In terms of a potential $U(x)$ and a work
function (or driving) $F(x,y) = -F(y,x)$, the rates now satisfy
\begin{equation}\label{ldb}
\lambda(x,y) = e^{\beta\,[F(x,y)+U(x)-U(y)]}\,\lambda(y,x)
\end{equation}
where $\beta\geq 0$ can still be interpreted as the inverse
temperature of a reference reservoir, but that is not necessary
except for setting the right units.

We can rewrite the local detailed balance condition \eqref{ldb} as
\begin{equation}\label{transrates}
\rho_0(x)\lambda(x,y) = \gamma(x,y)\,e^{\frac{\beta}{2}F(x,y)}
\end{equation}
for a symmetric $\gamma(x,y) = \gamma(y,x)$, which here is
arbitrary.

In the present paper we are concerned with the
close-to-equilibrium regime where $F$ is small.  To make it
precise, we parameterize the distance to equilibrium explicitly by
assuming that $\ga(x,y) = \ga_\ve(x,y)$ and $F(x,y) = F_\ve(x,y)$
(and hence also $\la(x,y) = \la_\ve(x,y)$) depend on a parameter
$\ve \in [0,\ve_0]$. With no loss of generality we let
$F_\ve(x,y) = \varepsilon F_1(x,y)$. As becomes obvious later,
the $\ve-$dependence of $\ga_\ve(x,y)$ is irrelevant for the first-order calculations.\\

We also consider the probability of trajectories, or rather, how
to obtain probability densities in path-space.  For this we start
with a probability law $\mu$ at time zero for the Markov process,
and write $\cal{P}_{\mu}$ for its  path-space distribution over a
time interval $[0,T]$. That has a density with respect to the
corresponding stationary equilibrium process $\cal{P}_0$ with
rates $\lambda_0(x,y)$ and starting from $\rho_0$, explicitly
given by the Girsanov formula
\begin{equation}\label{girs}
\frac{\id\caP_{\mu}}{\id \caP_0}(\omega) = \frac{\mu(x_0)}{
\rho_0(x_0)}\exp \Bigl\{ -\int_0^T \, \bigl(\xi(x_t) -  \xi_0(x_t)
\bigl)\,\id t  + \sum_{0< t\leq T}\log
\frac{\lambda(x_{t^-},x_{t})}{{\lambda_0}(x_{t^-},x_{t})} \Bigr\}
\end{equation}
where $\omega= (x_t)_{t=0}^T$, $x_t\in \Omega,$ is a piecewise
constant right-continuous trajectory, with the escape rates
\[ \xi(x) = \sum_{y \neq x} \lambda(x,y)\]
and with the last sum in the exponent being over the jump times
$t$ where the state changes from $x_{t^-}$ to $x_t$. Mathematical
details are found in, e.g., Appendix~2 of~\cite{KL}.

From \eqref{girs}, the path-space action $A$ in
\begin{equation}\label{girsan}
\id\cal{P}_\mu(\omega) =
\id\cal{P}_0(\omega)\,\frac{\mu(x_0)}{\rho_0(x_0)}\,e^{-A(\omega)}
\end{equation}
equals
\[
A(\omega) = \exp \Bigl\{ \int_0^T \, \bigl(\xi(x_t) -  \xi_0(x_t)
\bigl)\,\id t  - \sum_{0<t\leq T}\log
\frac{\lambda(x_{t^-},x_{t})}{{\lambda_0}(x_{t^-},x_{t})} \Bigr\}
\]
As a result, its time-antisymmetric part is
\begin{eqnarray}\label{sir}
S_{\irr}^T(\omega) &=& A(\theta\omega) - A(\omega)\nonumber\\
&=& \beta \sum_{0<t\leq T} F(x_{t^-},x_{t})
\end{eqnarray}
where the time-reversal is the right-continuous modification of
$\theta \omega = (x_{T-t})_{t=0}^T$ for any $\omega =
(x_t)_{t=0}^T$. We have used that the first integral in the
exponent of~\eqref{girs} is time-symmetric, and that
\eqref{db}--\eqref{transrates} combine to produce the forcing in
the sum over jump times. We recognize the resulting
$S^T_{\irr}(\omega)$ as the `irreversible' part in the entropy
flux as function of the path $\omega$. Note that the pathwise
relation~\eqref{sir} between the time-reversal symmetry breaking
and the entropy flux is a consequence of  condition~\eqref{ldb}.
For
general arguments see e.g.~\cite{jmp2}.\\

The mean value of that irreversible part of the entropy production
is obtained by taking the average of~\eqref{sir} with respect to
our process, using its Markov property: 
\begin{equation}\label{sout}
\begin{split}
  \langle S_{\irr}^T \rangle_\mu &=
  \int\id\cal{P}_\mu(\omega)\,S_{\irr}^T(\omega) = \int_0^T \id t\,
  \lim_{\tau \downarrow 0} \frac{1}{\tau}
  \langle S_{\irr}^\tau \rangle_{\mu_t}
\\
  &=  \be\int_0^T \id t\, \sum_{x} \mu_t(x)
  \sum_{y \neq x} \la_{\varepsilon}(x,y)\, F_\ve(x,y)
  \\
  &=  \be\int_0^T \id t\, \Bigl\langle
  \sum_{y \neq x} \la_\varepsilon(x_t,y)\, F_\ve(x_t,y)\Bigr\rangle_{\mu}
\end{split}
\end{equation}
Hence, for fixed $T$ we have to first order in $\varepsilon$, 
\begin{equation}\label{w10}
  \langle S_{\irr}^T \rangle_\mu = \varepsilon\be\,\int_0^T \id t\,
  \bigl\langle w_1(x_t)\bigr\rangle_{\mu}^0 + O(\varepsilon^2)
\end{equation}
where the averaging $\langle\cdot\rangle^0_{\mu}$ is now over the
equilibrium reference process started from $\mu$, and 
\begin{equation}\label{w1}
  w_1(x) = \sum_{y \neq x} \la_0(x,y)\,F_1(x,y)
\end{equation}
is the linear term in the mean entropy flux when at state $x$.

\subsection{McLennan formula}\label{der}

To be explicit about the various dependencies, we write
$\rho_T^\varepsilon$ for the $\varepsilon-$dependent solution at
time $T$ to the Master equation \eqref{timev}, started from the
equilibrium law $\mu_0 =\rho_0$. The smoothness of the deformation
is assumed uniformly in time $T$, and we write $\rho_\ve = \lim_T
\rho_T^\varepsilon$. We also denote the stationary entropy flux by
$\si_\ve$; it is given as 
\begin{eqnarray}
\sigma_\varepsilon &=& \frac{1}{T} \langle S_{\irr}^T
  \rangle_{\rho_\ve} \nonumber\\
&=&
 \be \sum_{x} \rho_\ve(x) \sum_{y \neq x} \la_\ve(x,y) F_\ve(x,y)\nonumber\\
&=&
  \frac{\be}{2}\,
  \sum_{x,y} [\rho_\ve(x)\la_\ve(x,y) - \rho_\ve(y)\la_\ve(y,x)]\,F_\ve(x,y)\nonumber\\
&=&
  \frac{\varepsilon\be}{2}\,
  \sum_{x,y} \gamma_0(x,y)\,
  \Bigl[\frac{\rho_\ve(x)}{\rho_0(x)}-\frac{\rho_\ve(y)}{\rho_0(y)}\Bigr]\,
  F_1(x,y) + o(\varepsilon^2)
\end{eqnarray}
independently of time span $T$.  We finally recall the linear term
$w_1$ from \eqref{w1}.
\begin{theorem}\label{tt}
Suppose that the equilibrium process~\eqref{db} is irreducible.
The following limiting identities are verified: 
\begin{eqnarray}
\lim_{T\uparrow +\infty}\lim_{\varepsilon\rightarrow 0}\,\frac
1{\varepsilon}\log\frac{\rho_T^\varepsilon(x)}{\rho_0(x)} &=&
\lim_{\varepsilon\rightarrow 0}\lim_{T\uparrow +\infty}\,\frac
1{\varepsilon}\log\frac{\rho_T^\varepsilon(x)}{\rho_0(x)}\label{triv}\\
&=& -\be\int_0^{+\infty}\id t \,\langle w_1(x_t)\rangle_{x}^0
\label{1}
\end{eqnarray}
Moreover, 
\begin{equation}
  \varepsilon\be\int_0^{+\infty}\id t \,\langle w_1(x_t)\rangle_{x}^0
  = \lim_{T\uparrow +\infty} \big[\langle S_{\irr}^T \rangle_x
  - \sigma_\varepsilon T\big] + O(\varepsilon^2) \label{2}
\end{equation}
\end{theorem}
\begin{remark}
According to the above result, the stationary distribution has the
form 
\begin{equation}
\begin{split}
  \rho_\ve(x) &= \rho_0(x) \exp \Bigl\{-\varepsilon\be\,\int_0^{+\infty}
  \id t\,\langle w_1(x_t) \rangle^0_{x} + O(\ve^2) \Bigr\}
\\
  &= \rho_0(x) \exp \Bigl\{ \lim_{T \uparrow +\infty}
  \bigl[ \sigma_\varepsilon T - \langle S_{\irr}^T
  \rangle_x \bigr] + O(\varepsilon^2) \bigr\}
\end{split}
\end{equation}
which is consistent with the original McLennan's
proposal~\eqref{mac} in the sense that it
identifies the correction term $W$ as the transient part of the `irreversible' entropy production for the process started from state $x$. Notice that this transient part is $O(\ve)$, in contrast to the stationary entropy production rate which is $O(\ve^2)$; the latter being also the leading order of the long-time divergence that needs to be removed. Hence, loosely speaking, the McLennan proposal is all correct for close-to-equilibrium
processes provided that the divergence present in higher orders in $\ve$ is killed by a suitable counterterm.
\end{remark}
\begin{proof}
The first equality~\eqref{triv} follows from the irreducibility of
the reference equilibrium process. The $\ve-$dependent process is
obtained by its smooth deformation,
see~\eqref{ldb}--\eqref{transrates}. Therefore,
$\rho_T^\varepsilon \rightarrow \rho_\ve$ uniformly in
$\varepsilon \in [0,\varepsilon_1]$ with some $0 < \ve_1 \leq
\ve_0$.

In order to prove the equality~\eqref{1}, the point of departure
is the transient fluctuation symmetry. The formula~\eqref{girsan}
obviously implies
\[
\id\cal{P}_{\rho_0}(\omega) = \id\cal{P}_0(\omega)\,e^{-A(\omega)}
\]
when starting (in the left-hand side) the nonequilibrium process
from the equilibrium law $\rho_0$. Since the equilibrium process
$\cal{P}_0$ is time-reversal invariant, we have for \eqref{sir}
\[
S_{\irr}^T(\omega)  =
\log\frac{\id\cal{P}_{\rho_0}}{\id\cal{P}_{\rho_0}\theta}(\omega)
\] and hence, for all functions $f$ on path-space,
\begin{equation}\label{fluc}
 \langle f\rangle_{\rho_0} = \langle f\theta\;
\exp (-S_{\irr}^T) \rangle_{\rho_0}
\end{equation}
That is an exact (for all finite times $T$) fluctuation symmetry.
Take in \eqref{fluc} $f(\omega) = \delta_{x_T=x}$, the
Kronecker-delta function equal one if the trajectory ends up at
state $x$ and zero otherwise. We get
\begin{equation}\label{fs}
  \rho^\varepsilon_T(x) = \rho_0(x)\,\langle e^{-S_{\irr}^T}\rangle_x
\end{equation}
where the right-hand side averages over the nonequilibrium process
started from the state $x$. We substitute~\eqref{sir} and we use
that the Poisson number of jumps has all exponential moments, to
expand the exponential in~\eqref{fs}. Using~\eqref{sout}, it is
then easy to verify that 
\begin{equation}\label{the}
  \lim_{\varepsilon\downarrow 0} \frac 1{\varepsilon} \log \langle
  e^{-S_{\irr}^T}\rangle_x =
  - \be\Bigl\langle \int_0^T\id t \,w_1(x_t) \Bigr\rangle^0_x
\end{equation}
which has a limit as $T\uparrow +\infty$ since, uniformly in the
initial $x$, the equilibrium process relaxes exponentially fast to
its stationary law for which $\sum_x \rho_0(x) w_1(x) =0$ by
\eqref{db}. That proves formula \eqref{1}.

We now turn to \eqref{2}. For all initial laws $\mu$, 
\[
  \langle S_{\irr}^T \rangle_\mu - \sigma_\varepsilon T =
  \be \int_0^T \id t\, \sum_{x} [\mu_t(x)-\rho_\ve(x)]
  \sum_{y \neq x} \la_\ve(x,y) F_\ve(x,y)
\]
has a limit $T\uparrow +\infty$.  Since all terms are uniformly
bounded, we can here freely exchange the limits
$T\uparrow +\infty$ and $\varepsilon\downarrow 0$.  The $T-$limit has leading
order in $\varepsilon$, for $\varepsilon \downarrow 0$, 
\begin{eqnarray}\label{calcul}
  \varepsilon\be \int_0^{+\infty} \id t\, \sum_{x}
  [\mu_t^0(x)-\rho_0(x)] \sum_{y \neq x} \la_0(x,y)
  F_1(x,y)\nonumber\\
  = \varepsilon\be \int_0^{+\infty} \id t\, \sum_{x} \mu_t^0(x)
  \sum_{y \neq x} \la_0(x,y) F_1(x,y)
\end{eqnarray}
because $F_1(x,y)$ is antisymmetric and
$\rho_0(x)\lambda_0(x,y)$ is symmetric under $x\leftrightarrow y$.
That proves \eqref{2}.
\end{proof}

\subsection{Variations}

What is mostly new about the above arguments is the point of
departure~\eqref{fluc}, which naturally links McLennan's
correction to the `irreversible'
entropy production, cf. \cite{KN}. There are other schemes that do not have this advantage; still they can be applied to obtain systematic corrections to the equilibrium distribution. We discuss one such a formulation and we apply it to obtain a different perturbation scheme not having an interpretation in terms of the entropy production.\\

By using the fundamental theorem of calculus, the stationary
measure $\rho_\ve = \rho$ can be obtained in the form
\begin{equation}\label{ftc}
\rho(x) = \rho_0(x) + \int_0^{+\infty}\id t\,e^{t
L^\star}\,L^\star\rho_0(x)
\end{equation}
for initial law $\rho_0$ and with $L^\star\mu(x) = \sum_{y \neq x}
[\lambda(y,x)\mu(y) - \lambda(x,y)\mu(x)]$ the forward generator
of the Markov jump process under consideration. Then, from
\eqref{transrates},
\[
L^\star\rho_0(x) = \rho_0(x)\,h(x),\quad h(x) = \sum_{y \neq x}
\lambda(x,y) [e^{-\beta F(x,y)} - 1]
\]
is of order $\varepsilon$, and
\[
  e^{t L^\star_0}\,[\rho_0 h](x) = \rho_0(x) e^{t L_0}h(x)
\]
with $L_0$ the backward equilibrium generator. Substituting
to~\eqref{ftc}, we get the expression
\begin{equation}\label{ftc2}
  \rho_\ve(x) = \rho_0(x) \Bigl[ 1 + \int_0^{+\infty}\id t\,e^{t L_0}h(x)
  \Bigr] + O(\varepsilon^2)
\end{equation}
Since 
\[
  h(x) = -\ve\be \sum_{y \neq x}
  \lambda_0(x,y) F_1(x,y) + O(\varepsilon^2)
\]
we have rederived formula~\eqref{1}.\\

The scheme \eqref{ldb}--\eqref{transrates} in combination with the
$\varepsilon$-dependence according to 
$F(x,y) = \varepsilon F_1(x,y)$ is a special way of breaking the detailed balance
condition.  One easily encounters other physically relevant
mechanisms. For example, suppose that for the reference
equilibrium dynamics transitions $x \rightleftarrows y$ between
specific states are forbidden.  We may imagine two uncoupled
equilibrium systems.  The nonequilibrium dynamics could introduce
a small coupling with for example
\[
\lambda(x,y) = \varepsilon k(x,y),\;\; \lambda(y,x) = \varepsilon
k(y,x)
\]
for these specific transitions.  We do not longer enjoy then the
absolute continuity of the nonequilibrium process with respect to
the equilibrium reference and the relations
\eqref{girs}--\eqref{girsan} break down. The entropy production as
a function on path-space does not depend on $\varepsilon$.
Nevertheless, from \eqref{ftc} we can still compute the linear
correction to the reference law $\rho_0$.  It is exactly of the
form \eqref{ftc2} but with 
\begin{align}
h(x) &  = \varepsilon \sum_{y\nleftarrow x}
\Bigl[\frac{\rho_0(y)}{\rho_0(x)}k(y,x) - k(x,y)\Bigr] =
\varepsilon \sum_{y\nleftarrow x} k(x,y)[e^{-\phi(x,y)}
- 1]
\\ \intertext{and}
\phi(x,y) &= \log
\frac{\rho_0(x)k(x,y)}{\rho_0(y)k(y,x)}
\end{align}
where the sum is over the {\it for the equilibrium dynamics}
forbidden transitions from state $x$.  The difference with
nonequilibrium perturbations where one adds a small driving in a
local detailed balance condition, as in
\eqref{ldb}--\eqref{transrates}, is manifest. The correction to
equilibrium is not of the form of an entropy flux.  In other
words, not all perturbations from equilibrium, even physical ones,
lead to the same type of correction to the equilibrium
distribution: the specific McLennan correction in terms of the
irreversible entropy flux arises (only) by the change from
detailed balance to local detailed balance by inserting some small
driving.

\subsection{Example: boundary driven lattice gas}\label{leq}
The following example makes the above considerations and
formul{\ae} more concrete.
We also take the opportunity to explain the relation with local equilibrium.\\

We consider a lattice gas on the sites $\{-N,\ldots,0\}$ where the
configurations $x, y$ indicate the vacancy or presence $x(i) =
0,1$ of a particle at each site $i$. The dynamics distinguishes
two ways of updating.

We concentrate on the case where there is a bulk conservation law
as in Kawasaki dynamics.  To be specific we choose the bulk
transitions as an exchange of nearest neighbor occupancy: we write
$x^{i,i+1}$ for the configuration that equals $x$ except that
the occupations at sites $i$ and $i+1$ are interchanged.  Then,\\
\begin{equation}\label{star}
\lambda(x,y) = a_i\exp \bigl( -\frac{\beta}{2}\{U(y) - U(x)\}
\bigr) \text{ when } y=x^{i,i+1}, i=-N,\ldots,-1
\end{equation}
At the left and right boundaries, there is a birth and death
process: with $y^i$ the configuration for which at site $i$ the
occupation has been inverted,
\[
  \lambda(x,y) = \exp \bigl( -\frac{\beta}{2}\{U(y) - U(x)\} \bigr)\,
  \exp \bigl( -\frac{\beta b_i}{2} (2x(i)-1) \bigr) \text{ when } y=x^{i}, i=-N,0
\]
with $b_i$ playing the role of chemical potential of left and right reservoirs.\\
Suppose now that $a_i=1,b_{-N}=0, b_0=\varepsilon$ which is a
close-to-equilibrium system in the sense of Section \ref{s1}; the
equilibrium law is $\rho_0(x) \propto \exp[-\beta U(x)]$, reached
for $\varepsilon = 0$. The expression \eqref{w1} becomes 
\begin{equation}\label{w111}
  w_1(x) = \exp \bigl( -\frac{\beta}{2}\{U(x^0) - U(x)\} \bigr)\,
  [2x(0)-1]
\end{equation}
On the other hand, consider the function $g(x) = \sum_{i=-N}^{0}
v_i x(i)$, for some profile $v_i$.  Then,
\begin{equation}\label{lo}
\begin{split}
L_0 g(x) = &\sum_{i=-N}^{-1}  [v_i-v_{i+1}] \,j_{i}(x) +
v_0\,\exp[-\beta\{U(x^0) - U(x)\}][2x(0)-1]
\\
&+ v_{-N}\,\exp[-\beta\{U(x^{-N}) - U(x)\}][2x(-N)-1]
\end{split}
\end{equation}
with systematic currents over the bonds $(i,i+1)$
\[
 j_i(x) = \lambda(x,x^{i,i+1}) [x(i+1)-x(i)]
\]
as they appear in the continuity equation for local particle
number.  Choose 
$v_i = 1 + i/N$ which makes $v_0 = 1$, $v_{-N}=0$, and
$v_{i+1} - v_i = 1/N$. Comparing~\eqref{w111}
with \eqref{lo} yields 
\begin{equation}\label{los}
  L_0 g(x) = -\frac{1}{N}\sum_{i=-N}^{-1}j_{i}(x) + w_1(x)
\end{equation}
For the McLennan-form \eqref{1} we must take the time-integral of~\eqref{w111} so that 
\begin{equation}\label{loce}
\rho_\varepsilon(x)\propto \rho_0(x)\,\exp \Bigl\{\varepsilon\be
\sum_{i={-N}}^{0} v_i x(i) \Bigr\}\, \exp
\Bigl\{-\frac{\varepsilon\beta}{N}
\sum_{i=-N}^{-1}\int_0^{+\infty} \id t\,e^{tL_0}j_{i}(x) \Bigr\}
\end{equation}
This expression \eqref{loce} for the approximate stationary
distribution is of the form of local equilibrium for the conserved
quantity (particle number) containing the linear profile $v_i$ for
the (local) chemical potential.  The remaining integral
\[
\varepsilon\beta\,\int_0^{+\infty}\id
t\,\frac{1}{N}\sum_{i=-N}^{-1}\langle j_{i}(x_t)\rangle^0_x
\]
makes mathematical sense because the local currents die out
exponentially fast for the equilibrium dynamics. It was also
discussed in the same context as formula (3.49) in \cite{ELS}. As
$N$ gets large, the average current gets even smaller for each
fixed time.  It appears, without proof, that for boundary driven
spatially extended systems the McLennan-regime
close-to-equilibrium can also be reached by taking $N$ large, for
fixed chemical potential difference, while
that is not included in formulations such as~\eqref{1}.\\

\section{Diffusion processes}\label{dif}
In this section we treat Markov diffusion processes and we give an
alternative derivation of the main result.

\subsection{General argument}\label{sec: dif-gen}


We consider the class of $d-$dimensional inhomogeneous
It\^o-diffusions
\begin{equation}\label{gsd1}
  \id x_t = \bigl\{\chi(x_t) \bigl[ F(x_t) -
  \nabla U(x_t) \bigr] + \nabla\cdot D(x_t)\bigr\}\,\id t
  + \sqrt{2 D(x_t)}\, \id B_t
\end{equation}
defined on a torus (i.e., we assume periodic boundary conditions).
We use the notation $\nabla \cdot D$ for the vector with
components $\sum_k \partial_l D_{k l}$ and assume $D=\chi/\beta$;
the latter being the Einstein relation as a variant of the local
detailed balance condition for diffusion. The $d-$dimensional
vector $\id B_t$ has independent standard Gaussian white noise
components. We assume the fields $F(x)$, $U(x)$, and $\chi(x)$ are
smooth and the matrices $\chi(x)$ are symmetric and strictly
positive at all points $x$.

To each distribution with density $\mu$ there is associated a
current density
\begin{equation}\label{cu}
  j_\mu = \chi(F - \nabla U)\mu - D\nabla\mu
\end{equation}
and the stationarity of the law $\mu=\rho$ is equivalent with the
condition $\nabla \cdot j_\rho = 0$. The $j_\mu$  gives the
expected profile of the `real' particle current at given density
$\mu$ in the sense that  for any smooth function $f$,
\begin{equation}\label{str-identity}
  \Bigl\langle \int_0^T f(x_t) \circ \id x_t \Bigr\rangle_\mu
  = \int_0^T \id t \int f(x)\, j_{\mu_t}(x)\,\id x
\end{equation}
The left-hand side is the average of a Stratonovich-stochastic
integral  under the diffusion process started at time $t=0$ from
density~$\mu$. In particular, the instantaneous mean work of the
force $F$ is then
\begin{equation}\label{work}
  W(\mu) = \int F \cdot j_\mu\,\id x
  = \int \,w \mu\,\id x
\end{equation}
for \begin{equation}\label{workw}
  w = F \cdot \chi F - \chi F
  \cdot \nabla U + \nabla \cdot (D F)
\end{equation}
In order to check the McLennan-proposal we isolate the linear
order in \eqref{work}. We take the case $F = 0$, $\rho_0 =
\exp[-\be U]/Z > 0$ as the equilibrium reference and we expand
with small parameter $\varepsilon$:
\begin{eqnarray}\label{smo}
F &=& \varepsilon F_1 + \ldots\nonumber\\
w &=& \varepsilon w_1 + \ldots\\
\rho &=& \rho_\ve = \rho_0(1 + \varepsilon h_1 + \ldots)\nonumber
\end{eqnarray}
assuming smooth behavior around $\varepsilon=0$. From
\eqref{workw},
\begin{equation}\label{linwork}
  w_1 = \nabla \cdot (D F_1) - \chi F_1 \cdot \nabla U
  = \frac{\nabla \cdot (\rho_0 \chi F_1)}{\be\rho_0}
\end{equation}
the linear term in the mean work performed by the force $F$.  It
turns out, as in \eqref{1} and in agreement with McLennan's
proposal, that $h_1$ is given in terms of the mean total work
performed on the particle started from $x$ under the reference
dynamics.

\begin{theorem}\label{dd}
Suppose the process \eqref{gsd1} converges exponentially fast and
uniformly in initial states to its unique stationary probability
distribution with smooth density $\rho$ around $\rho_0$ as in
\eqref{smo}. Then,
\begin{equation}\label{ende}
  h_1(x)
  = -\be \Bigl\langle \int_0^\infty w_1(x_t)\,\id t \Bigr\rangle^0_x
\end{equation}
\end{theorem}

\begin{proof}
The current \eqref{cu} can be rewritten in terms of the reference
equilibrium density
\begin{equation}
  j_\mu = \mu\chi(F - \nabla U) - D\nabla\mu
  = \mu\chi F - \rho_0 D \nabla \bigl( \frac{\mu}{\rho_0} \bigr)
\end{equation}
so that the stationarity condition $\nabla\cdot j_\rho = 0$
implies
\begin{equation}\label{FP}
\begin{split}
  0
  &= \nabla \cdot \bigl[ \rho_0 D \nabla \bigl(\frac{\rho}{\rho_0}\bigr) - \rho\chi F \bigr]
\\
  &= \rho_0 D \nabla \cdot \nabla \bigl(\frac{\rho}{\rho_0}\bigr)
  + \rho_0 (\nabla \cdot D) \cdot \nabla \bigl(\frac{\rho}{\rho_0}\bigr)
  - \rho_0 \chi\nabla U \cdot \nabla \bigl(\frac{\rho}{\rho_0}\bigr)
  - \nabla \cdot (\rho\chi F)
\\
  &= \rho_0 L_0 \bigl(\frac{\rho}{\rho_0}\bigr) - \nabla \cdot (\rho\chi F)
\end{split}
\end{equation}
with $L_0$ the backward reference equilibrium generator.  To
linear order, that gives
\begin{equation}
  L_0 h_1 = \be w_1
\end{equation}
with $w_1$ from \eqref{linwork}.

As the relaxation to equilibrium is fast enough, $L_0$ can be
inverted on the space
\begin{equation}
  \Om^\bot = \{g:\,\int g \rho_0\,\id x = 0\}
\end{equation}
and since $w_1 \in \Om^\bot$, the stationary solution must have
first order
\begin{equation}
  h_1 = \be L_0^{-1} w_1
\end{equation}
or
\begin{equation}\label{end}
  h_1(x) = -\be \int_0^\infty \bigl( e^{t L_0} w_1 \bigr)(x) \,\id t
  = -\be \Bigl\langle \int_0^\infty w_1(x_t)\,\id t \Bigr\rangle^0_x
\end{equation}
as required.
\end{proof}

One might be tempted to write
\begin{equation} \label{linear}
  \frac{\rho_\ve}{\rho_0}(x) \stackrel{?}{=}
  1 - \be \Bigl\langle \int_0^\infty
  w(x_t)\,\id t \Bigr\rangle_x + O(\varepsilon^2)
\end{equation}
(i.e., with the \emph{full} work and possibly under the
nonequilibrium measure) but this is only formally true in the
sense that the $O(\varepsilon)$ terms on both sides are equal;
however the second term on the right-hand side generally diverges
now because of its $O(\varepsilon^2)$ part. Another way to see
that is by observing that $w \not\in \Om^\bot$ in general, in
which case $L_0^{-1}w$ does not exist.

Still, one can proceed similarly as in the case of jump processes
and add a suitable counter-term on the right-hand side
of~\eqref{linear}.  All that explains what is actually the
rigorous meaning of McLennan's proposal: to correctly describe the
first order correction, one is not allowed to deal with the
\emph{full} transient entropy production unless the divergences
coming from the high-order corrections are removed. For safe
first-order calculations one needs to take the linear part of the
entropy production functional only, as done above in
Theorem~\ref{dd}.

\subsection{Green-Kubo relations}\label{gk}

An expression for the close-to-equilibrium stationary density
obviously yields information about the stationary current in
linear response around equilibrium.
That  provides another derivation of the well-known Green-Kubo
relations between the current and equilibrium time-correlations.\\

We use the same notation as in the previous section but this time
we need to indicate the dependence on the driving force $F$, e.g.,
the linear part of the work~\eqref{linwork} is now written as
\begin{equation}
  w_1^F = 
  \frac{\nabla \cdot (\rho_0 \chi F_1)}{\be \rho_0}
\end{equation}
By expanding the mean stationary current~\eqref{cu} ($\mu = \rho$)
in powers of $\ve$, $j_\rho^F = \ve \frj_1^F + \ldots$, the
leading term has the form
\begin{equation}
  \frj_1^F = \rho_0 \chi\, [ F_1 - \nabla (L_0^{-1} w_1^F) ]
\end{equation}
where we have substituted the McLennan form. Suppose $G = \ve G_1
+ \ldots$ is another smooth field. Then we have, up to order
$\ve^2$, that $\int G \cdot j^F_\rho\,\id x = \ve^2 \int G_1 \cdot
\frj_1^F\,\id x + o(\ve^2)$ and
\begin{equation}\label{eq: G-F}
\begin{split}
  \int G_1 \cdot \frj_1^F\,\id x
  &= \int \rho_0 [G_1 \cdot \chi F_1 + \be w_1^G L_0^{-1} w_1^F]\,\id x
\end{split}
\end{equation}
Since $L_0$ is self-adjoint with respect to the scalar product
$(f,g) = \int \rho_0 \bar f g\,\id x$ (and $L_0^{-1}$ is therefore
self-adjoint on $\Om^\bot$) and as the matrix $\chi$ is symmetric,
we immediately get the Onsager reciprocity relations in the form
\begin{equation}
  \int G_1 \cdot \frj_1^F\,\id x =
  \int F_1 \cdot \frj_1^G\,\id x
\end{equation}
Note that for $F = G$ the formula~\eqref{eq: G-F} gives the
leading (= second order) term in the expansion for the stationary
instantaneous mean work $W^F(\rho^F)$, see \eqref{work}, whereas
for $F \neq G$ it corresponds to the `interference' contribution
when the driving fields are added:
\begin{equation}
  W^{F+G}(\rho^{F+G}) - W^F(\rho^F) - W^G(\rho^G) =
  2\ve^2 \int G_1 \cdot \frj_1^F\,\id x + o(\ve^2)
\end{equation}
\begin{theorem}\label{them: green-kubo}
Under the same assumptions as in Theorem~\ref{dd},
\begin{equation}\label{eq: green-kubo}
\begin{split}
  \int G_1 \cdot \frj_1^F\,\id x
  &= \int \id x \rho_0\,[\, G_1 \cdot \chi F_1 + \be w_1^G L_0^{-1}
  w_1^F \, ]
\\
  &= \lim_{T \uparrow +\infty} \frac{\be}{2T}
  \Bigl\langle \int_0^T G_1(x_t) \circ \id x_t
  \int_0^T F_1(x_s) \circ \id x_s \Bigr\rangle^0
\end{split}
\end{equation}
where $\circ$ indicates the Stratonovich integration
(incorporating the scalar product) and the last expectation is
under the equilibrium process.
\end{theorem}
\begin{remark}
The middle term in~\eqref{eq: green-kubo} is what follows from
applying the McLennan formula to the mean current
close-to-equilibrium.  The equality~\eqref{eq: green-kubo} then
yields the linear response formula for the close-to-equilibrium
stationary current in terms of current-current time correlations,
its right-hand side. The result can be formally summarized as
saying that
\begin{equation}
  j_\rho^F(x) = \int \caR(x,y)\,F(y)\,\id y + o(\ve)
\end{equation}
with a symmetric response function $\caR(x,y) = \caR(y,x)$ given
by
\begin{equation}
  \caR(x,y) = \lim_{T \uparrow +\infty} \frac{\be}{2T}
  \bigl\langle J_T(x)\, J_T(y) \Bigr\rangle^0
\end{equation}
and
\begin{equation}
  J_T(x) = \int_0^T \de(x_t - x) \circ \id x_t
\end{equation}
is the time-integrated empirical current density.
\end{remark}
\begin{proof}
As the first equality in~\eqref{eq: green-kubo} is
formula~\eqref{eq: G-F}, we only need to prove the second equality
there.

Using the standard relation between the Stratonovich and the It\^o
integrals, each integration along the equilibrium process
(corresponding to $F = 0$ in~\eqref{gsd1}) on the right-hand side
of~\eqref{eq: G-F} can be computed as follows:
\begin{equation}\label{itos}
\begin{split}
  \int_0^T F_1 \circ \id x_t &=
  \int_0^T F_1 \cdot \id x_t
  + \int_0^T \sum_{kl} D_{kl} \frac{\partial F_{1,l}}{\partial x_k}\,\id t
\\
  &= \int_0^T [ -F_1 \cdot \chi \nabla U +
  \nabla \cdot (D F_1) ]\,\id t +
  \int_0^T F_1 \cdot (2D)^{1/2} \id B_t
\\
  &= \int_0^T w_1^F \id t +
  \int_0^T F_1 \cdot (2D)^{1/2} \id B_t
\end{split}
\end{equation}
We must substitute \eqref{itos} into the right-hand side of
\eqref{eq: green-kubo}.  We start with the cross terms, one of
which is
\begin{equation}
\begin{split}
  \Bigl\langle \int_0^T w_1^G\,\id t &\int_0^T F_1 \cdot
  (2D)^{1/2}\id B_s \Bigr\rangle =
  \Bigl\langle \int_0^T w_1^G\,\id t \int_0^t F_1 \cdot
  (2D)^{1/2}\id B_s \Bigr\rangle^0
\\
  &= -\Bigl\langle \int_0^T w_1^G\,\id t
  \Bigl[\int_t^T F_1\cdot (2D)^{1/2} \id B_s + 2\int_t^T w_1^F\,\id s\Bigr]
  \Bigr\rangle^0
\\
  &= -2\Bigl\langle \int_0^T w_1^G\,\id t
  \int_t^T w_1^F\,\id s \Bigr\rangle^0
\end{split}
\end{equation}
In the second equality we have applied the time-reversal of the
last equality in \eqref{itos}, using that the
Stratonovich-integral is time-antisymmetric and that the
equilibrium process is time-reversal symmetric. Analogously,
\begin{equation}
\begin{split}
  \Bigl\langle \int_0^T w_1^F\,\id t &\int_0^T G_1 \cdot
  (2D)^{1/2}\id B_s \Bigr\rangle^0 =
  -2\Bigr\langle \int_0^T w_1^F\,\id t
  \int_t^T w_1^G\,\id s \Bigr\rangle^0
\\
  &= -2\Bigr\langle \int_0^T w_1^G\,\id t
  \int_0^t w_1^F\,\id s \Bigr\rangle^0
\end{split}
\end{equation}
Hence, both cross-terms together give
\begin{equation}
  -2\Bigl\langle \int_0^T w_1^G\,\id t
  \int_0^T w_1^F\,\id s \Bigr\rangle^0
\end{equation}
As a consequence, the correlation is then
\begin{equation*}
\begin{split}
  \Bigl\langle \int_0^T G_1(x_t) &\circ \id x_t
  \int_0^T F_1(x_s) \circ \id x_s \Bigr\rangle^0
\\
  &= 2\Bigl\langle \int_0^T G_1 \cdot D F_1(x_t)\,\id t \Bigr\rangle^0
  - \Bigl\langle \int_0^T w_1^G(x_t)\,\id t \int_0^T w_1^F(x_s)\,\id s
  \Bigr\rangle^0
\\
  &= 2T \langle G_1 \cdot D F_1 \rangle^0
  -  \int_{\al T}^{(1-\al)T} \id t \int_{-t}^{T-t} \id s\,
  \langle w_1^G(x_t)\,w_1^F(x_s) \rangle^0 + O(\al T)
\\
  &= 2T \langle G_1 \cdot D F_1 \rangle^0
  - (1 - 2\al)T \int_{-\infty}^{+\infty}
  \id s\,\langle w_1^G(x_0)\,w_1^F(x_s) \rangle^0 + O(\al T) +
  O(T\,e^{-\ka \al T})
\end{split}
\end{equation*}
for an arbitrary $0 < \al < 1/2$ and with $\ka$ the rate of
exponential decay of the time correlations. Dividing by $2T$ and
taking the limits $T \uparrow +\infty$ and then $\al \downarrow
0$, one finally obtains
\begin{equation}
  \lim_{T \uparrow +\infty} \frac{1}{2T}
  \Bigl\langle \int_0^T G_1(x_t) \circ \id x_t
  \int_0^T F_1(x_s) \circ \id x_s \Bigr\rangle^0
  = \langle G_1 \cdot D F_1 \rangle^0
  - \int_0^{+\infty} \id s \, \langle w_1^G\,e^{s L_0} w_1^F \rangle^0
\end{equation}
\end{proof}

\begin{remark}
The same proof applies for jump processes.  One then has the
analogues $w_1^F=\sum_y \lambda_0(x,y) F(x,y),\;F(x,y) = -
F(y,x)$, and   $\frj_1^F = \gamma(x,y)\,[F(x,y) + h^F(x) -
h^F(y)],\; \gamma(x,y) =\rho_0(x)\lambda_0(x,y)=\gamma(y,x),
h^F(x) = L_0^{-1} w^F_1 (x),\; \rho^F(x) = \rho_0(x) [1 +
h^F(x)]$. Theorem~\ref{them: green-kubo} applies equally under the
condition of uniform exponential relaxation.
\end{remark}

\subsection{Even and odd variables: example}\label{odd}

For models whose configurations do not transform trivially under
time reversal, the above construction requires a generalization.
If the involution $\pi$, $\pi^2 = \funit$, is the kinematical
time-reversal on the state space, the detailed balance condition
takes the generalized form
\begin{equation}\label{eq: db-gen}
  L_0^+ = \pi L_0 \pi,\qquad \rho_0\pi = \rho_0
\end{equation}
for the adjoint $L_0^+$ in the sense $\int \bar{f} (L_0
g)\,\rho_0\id x = \int \overline{(L_0^+ f)} g\, \rho_0\id x$.
Typical examples are dissipative mechanical systems, e.g.~underdamped diffusion processes, with states given by both
position and momentum variables for which $\pi$ turns the sign of
the momentum. Models of heat conduction are a prime example, and
the close-to-equilibrium analysis would imply Fourier's law --- if
indeed the equilibrium time correlation functions can be
controlled in the Green-Kubo formula, which remains highly
non-trivial, see e.g.~\cite{fou,ban}.

Since the McLennan form is fundamentally a consequence of the transient fluctuation symmetry, see Section~\ref{der}, the above arguments need only minor changes. Instead of repeating the whole derivation, we restrict ourselves to a simple example that elucidates essential points.\\

We consider the model of a linear RLC-circuit with two resistors
in series and one external voltage ($E$);  cf.~Section~3.3
in~\cite{BMN}. The two independent free variables are the
potential $U$ (even) over the first resistance $R_1$ and the
current $I$ (odd) through the second resistance $R_2$. To have a
nontrivial transient regime, an inductance $L$ is added in series
with the resistors and also a capacitance $C$ is connected in
parallel with the first resistor.  The environment is at inverse
temperature $\beta$.

The stochastic dynamics for $(U,I)$ is given by Kirchoff's laws
combined with the Johnson-Nyquist theory according to which
resistances $R_{1,2}$ immersed in a thermal environment are source
of an extra random voltage; these can be modeled as independent
Brownian motions with variance $R_{1,2} \be^{-1}$ per unit time.
Altogether,
\begin{equation}\label{rrcllang}
\begin{split}
  C\id U_t &= \Bigl( I_t - \frac{U_t}{R_1} \Bigr)\,\id t +
  (\be R_1)^{-1/2} \,\id B_{1,t}
\\
  L \id I_t &= (E - R_2 I_t - U_t)\,\id t + (\be R_2)^{-1/2} \,\id B_{2,t}
\end{split}
\end{equation}
The reference equilibrium dynamics corresponds to $E=0$, for which
the stationary density reads $\rho_0(U,I) \propto \exp \{-(C U^2 +
L I^2)/2 \} $ and the detailed balance condition~\eqref{eq:
db-gen} is verified with $\pi(U,I) = (U,-I)$.
 In the driven case, $E \neq 0$,
  and by the linearity of the example it is again easy to compute the stationary density
   $\rho_E$.  Instead we use this example to illustrate and to
    verify the McLennan proposal~\eqref{mac}.\\

The goal is to get $\rho_E$ from physically identifying the
transient entropy flux. The irreversible part of the dissipation
here equals $\be$ times the work done by the battery as a function
of the trajectory $\om = \{(U_t,I_t), t\in [0,T] \}$.
Specifically,
\begin{equation}\label{sira}
  S_{\irr}^T(\omega) = \be E \int_0^T I_t\,\id t
\end{equation}
Of course, when in doubt, one can compute it also by the general
algorithm as the time-antisymmetric part of the logarithmic
density of the path-measure with respect to its time-reversal,
cf.~\eqref{girs}--\eqref{w1}. By the linearity and homogeneity of
the system, the source voltage $E$ takes over the driving
parameter $\ve$ of the previous analysis and, using the notation
from there, the expected irreversible part of the dissipation is
\begin{equation}
  \langle S_{\irr}^T \rangle_\mu =
  E \int_0^T \langle w_1(I_t) \rangle_\mu^0\,\id t,\qquad
  w_1 = \be I
\end{equation}
cf.~\eqref{w10}--\eqref{w1} or \eqref{work}--\eqref{workw}. It is
now easy to check that the stationary density satisfies the
identities
\begin{equation}
\begin{split}
  \log\frac{\rho_E}{\rho_0}(U,I) &=
  -E \Bigl\langle \int_0^{+\infty} w_1(x_t)\,\id t \Bigr\rangle_{U,I}^{0,+}
  + O(E^2)
\\
  &= \be E (L_0^+)^{-1} I + O(E^2)
\\
  &= \be E \Bigl(\frac{R_1 C}{R_1 + R_2} \,U +
  \frac{L}{R_1 + R_2}\,I \Bigr) + O(E^2)
\end{split}
\end{equation}
where the equilibrium dynamics on the first line runs according to
the generator $L_0^+$, i.e., backwards in time. Since it is
identical to $\pi L_0 \pi$, it only generates sign changes in the
computation---that is the only point in which the previous
analysis must be modified.  The result of the McLennan theory thus
gives $\rho_E$ correctly as the Gaussian density with the same
covariance matrix
 as $\rho_0$ but with nonzero averages
\begin{equation}\label{gaus}
  \langle I \rangle = \frac{E}{R_1 + R_2}\, \qquad
  \langle U \rangle = \frac{R_1 E}{R_1 + R_2}
\end{equation}

The point is that even in cases where computations would be more
involved, the McLennan formula is in terms of a physical quantity
that can often be written down without the need to go much into
further details of the model.

\section{Generalization beyond close-to-equilibrium}\label{max}

It is natural to ask whether similar representations of the
stationary distribution remain valid also beyond
close-to-equilibrium.  The correction to quadratic order has been
systematically explored in \cite{KN,sa,ha}, within the programme
of steady state thermodynamics. Here we add a general expression,
\eqref{dent} below, from which a cumulant expansion around equilibrium could be started in principle.\\

We look back at \eqref{girs}--\eqref{girsan} that we now write as
\[
\id\cal{P}_\mu(\omega) = \id\cal{P}_\mu^0(\omega)\,e^{-A(\omega)}
\]
where the equilibrium reference ${\cal P}_\mu^0$ is starting from
the law $\mu$. We decompose the action $A$ into a
time-antisymmetric and a time-symmetric part:
\[
S = A\theta - A,\qquad {\cal T} = A\theta + A
\]
where we have abbreviated \eqref{sir} to $S$.  From
\eqref{girsan}, the time-symmetric part ${\cal T}$ is the
time-integrated excess in escape rates for jump processes.  It is
more generally related to the dynamical activity in the process.
We have also called it traffic in the context of  dynamical
fluctuation theory, \cite{MN,mnw}. We thus have
\begin{equation}\label{gnet}
  \id\cal{P}_\mu = \id\cal{P}_\mu^0\,\exp
  \Bigl( \frac{S- \caT}{2} \Bigr),\quad
  \id\cal{P}_{\rho_0}\circ\theta =
  \id\cal{P}_{0}\,\exp \Bigl( -\frac{S+ \caT}{2} \Bigr)
\end{equation}
where the second identity uses the time-reversal invariance of the
equilibrium process (started at $\rho_0$). The integrated form of
the second identity reads
\[
  \langle f(\omega) \rangle_{\rho_0} =
  \langle f(\theta \omega)\,e^{- \frac{1}{2}(S(\omega) + {\caT}(\omega))}\rangle^0
\]
with the right-most expectation over the equilibrium process.  The
left expectation is for the nonequilibrium process starting in
$\rho_0$ so that, taking functions $f(\omega) = f(\omega_t)$ of
the state at a single time $t$, we get information about the
approach to the nonequilibrium stationary density.  For the finite
state space $\Omega$, taking $f(\omega) = \delta_{x_T = x}$, we
get
\begin{equation}\label{dent}
  \text{Prob}(x_T = x) = \rho^\ve_T(x) =
  \rho_0(x)\,\langle e^{-\frac{1}{2} (S+{\caT})} \rangle_x^0,\quad
  x\in \Omega
\end{equation}
which is an exact formula for the nonequilibrium density at time
$T$, no matter
how far from equilibrium, entirely in terms of the reference
equilibrium dynamics starting at $x$.  Recall that the exponent in
the average is extensive in (large) time $T$ while also of (small)
order $\varepsilon$ around equilibrium. Together with the
normalization condition
\begin{equation}\label{norm}
  \langle e^{\frac{1}{2}(S - \caT)} \rangle^0_\mu = 1
\end{equation}
for all initial laws $\mu$, as follows from the first identity in
\eqref{gnet}, this can be taken as the starting point for
systematic expansions. In the first order around equilibrium, the
expected entropy flux $S$ and the expected traffic ${\cal T}$ are
identical. That explains how~\eqref{fs} and hence  McLennan's
formula follow from~\eqref{dent}.
That finally is why we call~\eqref{dent} a  generalization, still
involving the irreversible entropy flux $S$ but now also the
traffic $\caT$. Note that all expectations
in~\eqref{dent}--\eqref{norm} are under the equilibrium process,
in contrast to a slightly different construction suggested
in~\cite{KN} which only uses the variable entropy production but
with respect to the full nonequilibrium dynamics.  We see that
difference, equilibrium versus nonequilibrium expectation, also
when comparing~\eqref{dent} with~\eqref{fs} where $S_{\irr}^T =
S$:
\[
 \langle e^{-S_{\irr}^T}\rangle_x =
\langle e^{-\frac{1}{2} (S+{\caT})} \rangle_x^0
\]
Yet, starting from either of the two expressions, there remains
the difficulty of writing down a general and physically meaningful
expression of the stationary distribution $\rho$ in a sufficiently
explicit way,  also because of generic nonlocal aspects in the
relation between potential and stationary density, given a
nonequilibrium driving, \cite{MN,mnw}.

\begin{remark}
The expression \eqref{dent} or the original McLennan formula
\eqref{mac} is perhaps related to what is argued to follow  from a
maximum entropy principle, cf. \cite{Dewar,evans}. Remark however
that finding the correct constraints or observables for which to
apply such a principle remains highly unclear.   In particular, it
appears that certainly beyond linear order around equilibrium,
also the time-symmetric fluctuation sector must get involved, cf.
\cite{maarten}.
\end{remark}

\section{Conclusions}
McLennan's formula gives a Gibbsian-like expression for the steady
law close-to-equilibrium under the condition of local detailed
balance.   The correction to equilibrium involves the transient
entropy flux.  It was seen before in \cite{KN} how that arises
from a transient fluctuation formula.  We have added mathematical
precision in the order of limits (time versus distance from
equilibrium).  We have shown how it relates to local equilibrium
and to the Green-Kubo relations and we have presented a
generalization involving the dynamical activity.

In our opinion it remains important to attempt a thermodynamic
interpretation of also higher order corrections to equilibrium.

\begin{acknowledgements}
The authors thank T.~S.~Komatsu, N.~Nakagawa, S.~Sasa, and
H.~Tasaki for very fruitful discussions. K.N.\ acknowledges the
support from the Grant Agency of the Czech Republic (Grant
no.~202/07/0404). C.M.\ benefits from the Belgian Interuniversity
Attraction Poles Programme P6/02.
\end{acknowledgements}

\end{document}